\documentclass[10pt]{llncs}
\usepackage[]{lineno}

\usepackage{amssymb}
\usepackage{amsmath}

\usepackage{hyperref}
\usepackage{wrapfig} 
\usepackage{stmaryrd} 
\usepackage{paralist}
\usepackage{algorithm}
\usepackage[noend]{algpseudocode}
\usepackage{todonotes}
\usepackage{hyperref}
\usepackage[T1]{fontenc}
\usepackage{colortbl}

\usepackage{multirow}

\usepackage[T1]{fontenc}

\usepackage{graphicx}
\usepackage{subcaption}

\usepackage{listings}

\usepackage{color}
\definecolor{lightgray}{rgb}{.9,.9,.9}
\definecolor{darkgray}{rgb}{.4,.4,.4}
\definecolor{purple}{rgb}{0.65, 0.12, 0.82}
\lstdefinelanguage{TLSF}{
  keywords={INFO, TITLE, DESCRIPTION, SEMANTICS, TARGET, MAIN, INPUTS, OUTPUTS, ASSUMPTIONS, GUARANTEES, ASSERT, INITIALLY, REQUIRE, PRESET, ASSUME, GLOBAL, PARAMETERS, DEFINITIONS, IN},
  keywordstyle=\color{blue}\bfseries,
  ndkeywords={class, export, boolean, throw, implements, import, this},
  ndkeywordstyle=\color{darkgray}\bfseries,
  morekeywords={G,X},
  identifierstyle=\color{black},
  sensitive=false,
  comment=[l]{//},
  morecomment=[s]{[}{]},
  commentstyle=\color{red}\ttfamily,
  stringstyle=\color{red}\ttfamily,
  morestring=[b]',
  basicstyle=\small\ttfamily,
  morestring=[b]"
}

\lstdefinelanguage{Haskell}{
    basicstyle=\ttfamily\scriptsize,
    keywordstyle=\color{blue}\bfseries,
    commentstyle=\color{gray},
    stringstyle=\color{olive},
    morekeywords={if, then, else, let, in, case, of, where, do, module, import, qualified},
    sensitive=true,
    morecomment=[l]{--},
    morecomment=[s]{\{-}{-\}},
    morestring=[b]",
    morestring=[b]',
}

\input{macros}

\newcommand{\G}{\LTLsquare}
\newcommand{\X}{\LTLcircle}
\newcommand{\U}{\mathbin{\mathcal{U}}}

\newcommand{\contract}{S }
\newcommand{\contractfull}{S = (A,G)}

\newcommand{\TRUE}{\textit{True}\xspace}
\newcommand{\FALSE}{\textit{False}\xspace}

\newcommand{\changing}[1]{#1}

\newcommand{\LTL}{\textrm{LTL}\xspace}
\newcommand{\LTLX}{\ensuremath{\LTL_X}\xspace}
\newcommand{\IniE}{\ensuremath{I_e}}
\newcommand{\IniS}{\ensuremath{I_s}}
\newcommand{\AssE}{\ensuremath{\varphi_e}}
\newcommand{\AssS}{\ensuremath{\varphi_s}}
\newcommand{\ASF}{\KWD{ASF}}

\newcommand{\GXz}{\ensuremath{\textsf{GX}_0}\xspace}
\usepackage{tikz}
\usetikzlibrary{automata, positioning}
\usetikzlibrary{shapes.geometric}
\usetikzlibrary{arrows.meta,arrows}
\def\orcidID#1{\smash{\href{http://orcid.org/#1}{\protect\raisebox{-1.25pt}{\protect\includegraphics{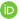}}}}}
  
\begin{document}

\setcounter{page}{1}

 \newcommand{\Thanks}{\thanks{Funded by PRODIGY Project 
  (TED2021-132464B-I00)---funded by MCIN/AEI/ 10.13039/501100011033 and
  the EU NextGenerationEU/PRTR---,by DECO Project
  (PID2022-138072OB-I00)---funded by MCIN/AEI/10.13039/501100011033 and
  by the ESF+---and by a research grant from Nomadic Labs and the Tezos
  Foundation.}}

\title{Efficient Reactive Synthesis\\ Using Mode Decomposition\Thanks}

\author{Mat\'ias Brizzio\inst{1,2}\orcidID{0009-0000-9427-9345}\and C\'esar S\'anchez\inst{1}\orcidID{0000-0003-3927-4773}}
\institute{IMDEA Software Institute, Spain \and
  Universidad Politécnica de Madrid, Madrid, Spain}






\maketitle

\begin{abstract}
  Developing critical components, such as mission controllers or
  embedded systems, is a challenging task.
  Reactive synthesis is a technique to automatically produce correct
  controllers.
  Given a high-level specification written in LTL, reactive synthesis
  consists of computing a system that satisfies the specification as
  long as the environment respects the assumptions.
  Unfortunately, LTL synthesis suffers from high computational
  complexity which precludes its use for many large cases.

  A promising approach to improve synthesis scalability consists of
  decomposing a safety specification into a smaller specifications,
  that can be processed independently and composed into a solution for
  the original specification.
  Previous decomposition methods focus on identifying independent
  parts of the specification whose systems are combined via
  simultaneous execution.
  
  In this work, we propose a novel decomposition algorithm based on
  \emph{modes}, which consists on decomposing a complex safety
  specification into smaller problems whose solution is then composed
  \emph{sequentially} (instead of simultaneously).
  The input to our algorithm is the original specification and the
  description of the modes.
  We show how to generate sub-specifications automatically and we
  prove that if all sub-problems are realizable then the full
  specification is realizable.
  Moreover, we show how to construct a system for the original
  specification from sub-systems for the decomposed specifications.
  We finally illustrate the feasibility of our approach with multiple
  cases studies using off-the-self synthesis tools to process the
  obtained sub-problems.
\end{abstract}


\section{Introduction}
\label{sec:intro}

Reactive synthesis~\cite{Church1962LogicAA} is the problem of
constructing a reactive system automatically from a high-level
description of its desired behavior.
A reactive system \changing{interacts} continuously with an uncontrollable
external environment~\cite{church1963application}.
The specification describes both the assumptions that the environment
is supposed to follow and the goal that the system must satisfy.
Reactive synthesis guarantees that every execution of the system
synthesized satisfies the specification as long as the environment
respects the assumptions.

Linear-Time Temporal Logic (LTL)~\cite{pnueli77temporal} is a widely
used formalism in verification~\cite{manna95temporal} and
synthesis~\cite{PnueliRosner1989} of reactive systems.
Reactive synthesis can produce controllers which are essential for
various applications, including hardware design~\cite{bloem12synthesis}
and control of autonomous robotic
systems~\cite{kress2011correct,dippolito13synhesizing}.

Many reactive synthesis tools have been developed in recent
years~\cite{finucane2011designing,ehlers2016slugs} in spite of the
high complexity of the synthesis problem.
Reactive synthesis for full LTL is
2EXPTIME-complete~\cite{PnueliRosner1989}, so LTL fragments with
better complexity have been identified.
For example, GR(1)---general reactivity with rank 1---enjoys an
efficient (polynomial) symbolic synthesis
algorithm~\cite{bloem12synthesis}.
Even though GR(1) can express the safety fragment of LTL considered in
this paper, translating our specifications into GR(1) involves at
least an exponential blow-up in the worst case~\cite{hermo23tableaux}.
Better scalable algorithms for reactive synthesis are still
required~\cite{kupferman2012recent}.

Model checking, which consists on deciding whether a \emph{given
  system} satisfies the specification, is an easier problem than
synthesis.
Compositional approaches to model checking break down the analysis
into smaller sub-tasks, which significantly improve the performance.
Similarly, in this paper we aim to improve the scalability of reactive
synthesis introducing a novel decomposition approach that breaks down
the original specification into multiple sub-specifications.

There are theoretical compositional approaches
\cite{esparza2014ltl,kupferman2006safraless}, and implementations that
handle large
conjunctions~\cite{bansal2020hybrid,de2021compositional,meyer2018strix}.
For instance, Lisa~\cite{bansal2020hybrid} has successfully scaled
synthesis to significant conjunctions of LTL formulas over finite
traces (a.k.a. LTL$_f$\cite{de2013linear}).
Lisa is further extended to handling prominent disjunctions in
Lydia~\cite{de2021compositional}.
These modular synthesis approaches rely heavily on the decomposition
of the specification into simultaneous
sub-specifications~\cite{finkbeiner2022specification}.
However, when sub-specifications share multiple variables, these
approaches typically return the exact original specification, failing
to generate smaller decompositions.

We tackle this difficulty by introducing a novel decomposition
algorithm for safety LTL specifications.
We chose the safety fragment of
LTL~\cite{sistla1994safety,kupferman01model} because it is a
fundamental requirement language in many safety-critical applications.
Extending our approach to larger temporal fragments of LTL is future
work.

To break down a specification we use the concept of \textit{mode}.
A mode is a sub-set of the states in which the system can be during
its execution which is of particular relevance for the designer of the
system.
At any given point in the execution, the system is in a single mode,
and during an execution the system can transition between modes.
In requirement design, the intention of modes is often explicitly
expressed by the requirement engineers as a \emph{high-level state
  machine}.
Using LTL reactive synthesis these modes are boiled down into
additional LTL requirements , which are then processed with the rest
of the specification.
In this paper, we propose to exploit modes to decompose the
specification into multiple synthesis sub-problems.

%
%
Most previous decomposition
methods~\cite{Ianopollo2018,finkbeiner2022specification} break
specifications into independent \emph{simultaneous} sub-specifications
whose corresponding games are solved independently and the system
strategies composed easily.
In contrast, we propose \emph{sequential} games, one for each mode.
For each mode decomposition, we restrict the conditions under which
each mode can ``jump'' into another mode based on the initial conditions
of the arriving mode.
From the point of local analysis of the game that corresponds to a
mode, jumping into another mode is permanently winning.
We show in this paper that our decomposition approach is
sound---meaning that given a specification, system modes and initial
conditions---if all the sub-specifications generated are realizable,
then the original specification is realizable.
Moreover, we show a synthesis method that efficiently constructs a
system for the full specification from systems synthesized for the
sub-specifications.
An additional advantage of our method is that the automaton that
encodes the solution is structured according to the modes proposed, so
it is simpler to understand by the user.

\subsubsection{Related Work.}
%
The problem of reactive synthesis from temporal logic specifications
has been studied for many
years~\cite{EmersonClarke1982,PnueliRosner1989,AlurLaTorre2001,bloem12synthesis}.
%
%
Given its high complexity (2EXPTIME-complete~\cite{PnueliRosner1989})
easier fragments of LTL have been studied.
For example, reactive synthesis for GR(1) specifications can be solved
in polynomial time~\cite{bloem12synthesis}.
Safety-LTL has attracted significant interest due to its algorithmic
simplicity compared to general LTL synthesis~\cite{zhu2017symbolic},
but the construction of deterministic safety automaton presents a
performance bottleneck for large formulas.

For the model-checking problem, compositional approaches improve the
scalability significantly~\cite{roever1998}, even for large formulas.
Remarkably, these approaches break down the analysis into smaller
sub-tasks~\cite{PnueliRosner1989}.
For model-checking, Dureja and Rozier~\cite{dureja2018more} propose to
analyze dependencies between properties to reduce the number of
model-checking tasks.
Recently, Finkbeiner et al.~\cite{finkbeiner2022specification} adapt
this idea to synthesis, where the dependency analysis is based on
controllable variables, which makes the decomposition impossible when
the requirements that form the specification share many system
(controlled) variables.
We propose an alternative approach for dependency analysis in the
context of system specification, by leveraging the concept of
\textit{mode} to break down a specification into smaller components.
This approach is a common practice in Requirements Engineering
(\textit{RE})\changing{~\cite{heitmeyer1995consistency,heitmeyer2011requirements}
  where specifications typically contain a high-level state machine
  description (where states are called modes) and most requirements
  are specific to each mode.
Furthermore, this approach finds widespread application in various
industries, employing languages such as \emph{EARS}~\cite{EARS} and
\emph{NASA}'s \emph{FRET} language~\cite{FRET}.}
Recently, a notion of \textit{context} is introduced by Mallozi et
al~\cite{mallozzi2022contractbased} in their recent work on
\textit{assume-guarantee} contracts.
Unlike modes, contexts depend solely on the environment and are not
part of the elicitation process or the system specification.

Software Cost Reduction
(SCR)~\cite{heninger1978software,heitmeyer1995consistency,heitmeyer2011requirements}
is a well-establish technique that structures specifications around
\textit{mode classes} and \textit{modes}.
A mode class refers to internally controlled variables that maintain
state information with a set of possible values known as modes.

We use modes here provided by the user to accelerate synthesis,
exploiting that in RE modes are comonly provided by the engineer
during system specification.
Recently, Balachander et al.~\cite{balachander23ltl} proposed a method
to assist the synthesis process by providing a sketch of the desired
Mealy machine, which can help to produce a system that better aligns
with the engineer's intentions.
This approach is currently still only effective for small systems, as
it requires the synthesis of the system followed by the generation of
example traces to guide the search for a reasonable solution.
In contrast our interest is in the decomposition of the synthesis
process in multiple synthesis sub-tasks.

Other compositional synthesis approaches aim to incrementally add
requirements to a system specification during its
design~\cite{kupferman2006safraless}.
On the other hand,~\cite{filiot2010compositional}
and~\cite{finkbeiner2022specification} rely extensively on dropping
assumptions, which can restrict the ability to decompose complex
real-world specifications.
%




\section{Motivating Example}
\label{sec:running}

We illustrate the main ideas of our decomposition technique using the
following running example of a counter machine (\textit{CM}) with a
reset.
The system must count the number of ticks produced by an external
agent, unless the reset is signaled---also by the environment---in
which case the count is restarted.
When the count reaches a specific limit, the count has to be restarted as
well and an output variable is used to indicate that the bound has
been reached.
Fig.~\ref{lst:counter} shows a specification for this system with a
bound of $20$.
\changing{This example is written in TLSF (see~\cite{TLSF}), a
  well-established specification language for reactive synthesis,
  which is widely used as a standard language for the synthesis
  competition, \textit{SYNTCOMP}~\cite{SYNTCOMP}.}
\begin{figure}[b!]
\begin{lstlisting}[language=TLSF]
PARAMETERS { N = 20;}
INPUTS {reset;start;} OUTPUTS {counter[N+1];trigger;}
INITIALLY{ (!reset && !start);} ASSUMPTIONS{ G !(reset && start);}
PRESET{counter[0] && (&&[1 <= i <=N]!counter[i]);}
DEFINITIONS {
  mutual(b) = G ||[0 <= i < n](b[i] && &&[j IN {0, 1 .. (n-1)} (\) {i}] !b[j]);}}
GUARANTEES
  mutual(counter); G (reset -> X counter[0]);      
  G ((counter[0] &&  start) ->  X (counter[1] || reset));
  G ((counter[1] && !reset) ->  X (counter[2] || reset));
  ...
  G ((counter[N-1] && !reset) ->  X (counter[N] || reset));
  G (counter[N] -> X counter[0]);
  G (counter[N] -> trigger);  G (!counter[N] -> !trigger);
  \end{lstlisting}
  \caption{Counter machine specification.}
  \label{lst:counter}
\end{figure}
Even for this simple specification, all state-of-the-art synthesis
tools from the synthesis competition
\textit{SYNTCOMP}~\cite{SYNTCOMP}, including
\textit{Strix}~\cite{meyer2018strix}, are unable to produce a system
that satisfies \CM.

%
Recent decomposition
techniques~\cite{finkbeiner2022specification,Ianopollo2018} construct
a dependency graph considering controllable variable relationships,
but fail to decompose this specification due to the mutual
dependencies among output variables.
Our technique breaks down this specification into smaller
sub-specifications, grouping the counter machine for those states with
counter value $1$ and $2$ in a mode, states with counter $3$ and $4$
in a second mode, etc, as follows:

%
{
\centering
\includegraphics[scale=.90]{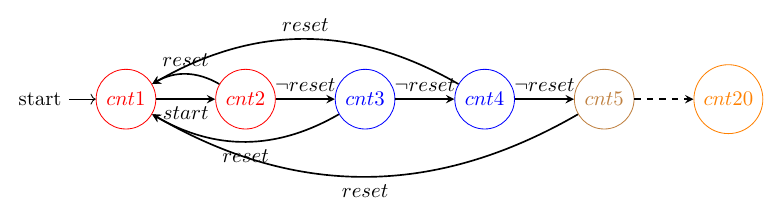}
}
%

\noindent Smaller controllers are synthesized independently, which can
be easily combined to satisfy the original
specification~\ref{lst:counter}.
In the example, we group the states in pairs for better readability,
but it is possible to use larger sizes.
In fact, for $N=20$ the optimal decomposition considers modes that
group four values of the counter (see Section~\ref{sec:empirical}).
The synthesis for each mode is efficient because in a given mode we
can ignore those requirements that involve valuations that belong to
other modes, leading to smaller specifications.

%
%

\begin{figure}[b!]
\begin{lstlisting}[language=TLSF, mathescape=true, label={lst:projections}]
// common part to all projections.
INPUTS {reset;start;} INITIALLY (!reset && !start); ASSUMPTIONS G !(reset && start);

[Projection under m1]
OUTPUTS {counter_0, counter_1; trigger; $\emph{jump}_2$; $s_{\X\varphi}$; $\emph{done}$}
GUARANTEES
  G (!$\emph{done}$ -> (counter_0 || counter_1));
  G (!$\emph{done}$ -> (reset -> X counter_0));
  G (!$\emph{done}$ -> (counter_0 && start) -> X (counter_1 || reset));
  G (!$\emph{done}$ -> ((counter_1 && !reset) -> $s_{\X\varphi}$));
  G (!$\emph{done}$ -> ($s_{\X\varphi}$ && !$\emph{done}$) -> X FALSE);
  G (!$\emph{done}$ -> !trigger);
  G ($\emph{done}$ -> X $\emph{done}$);
  G ($\emph{jump}_2$ -> X $\emph{done}$);
  G (!$\emph{jump}_2$ -> (!$\emph{done}$ -> X !$\emph{done}$));

  [Projection under m2]
  ...
  \end{lstlisting}
  \caption{Counter-Machine projection.}
  \label{lst:projections}
\end{figure}

In this work, we refer to these partitions of the state space as
modes.
In requirements engineering (\textit{RE}) it is common practice to
enrich reactive LTL specifications with a state transition system
based on modes, which are also used to describe many constraints that
only apply to specific modes.

Software cost reduction (SCR) uses modes in specifications and has
been successfully applied in requirements for safety-critical systems,
such as an aircraft's operational flight
program~\cite{heninger1978software}, a submarine's communication
system~\cite{heitmeyer1983abstract}, nuclear power
plant~\cite{van1993documentation}, among
others~\cite{bharadwaj1997applying,kirby1987example}.
SCR has also been used in the development of human-centric decision
systems~\cite{heitmeyer2015building}, 
and event-based transition
systems derived from goal-oriented requirements
models~\cite{letier2008deriving}.

%
Despite the long-standing use of modes in SCR, state-of-the-art
reactive synthesis tools have not fully utilized this concept.
The approach that we introduce in this paper exploits mode
descriptions to decompose specifications significantly reducing
synthesis time.
For instance, when decomposing our motivating example \CM using modes,
we were able to achieve 90\% reduction in the specification size,
measured as the number of clauses and the length of the specification
(see Section~\ref{sec:empirical}).
Fig.~\ref{lst:projections} shows the projections with a bound $N = 4$ for
mode $m_1 = (counter_0 \lor counter_1)$.
\changing{In each sub-specification, we introduce new variables
  (controlled by the system).
  These variables encode mode transitions using $\jump$ variables.
  When the system transitions to a new mode, the current
  sub-specification automatically wins the ongoing game, encoded by
  the \emph{done} variable
  A new game will start in the arriving mode.
  Furthermore, the system can only jump to new modes if the arriving
  mode is prepared, i.e., if its initial conditions---as indicated by
  the $s_{\X\varphi}$ variables---can satisfy the pending obligations.
  The semantics of these variables is further explained in the next
  section.}

%
%
In this work, we assume that the initial conditions are also provided
manually as part of the mode decomposition.
While modes are common practice in requirement specification, having
to manually provide initial conditions is the major current technical
drawback of our approach.
We will study in the future how to generate these initial conditions
automatically.
\changing{In summary, our algorithm receives the original
  specification $\contract$, a set of modes and their corresponding
  initial conditions.
  Then, it generates a sub-specification for each mode} and discharges
these to an off-the-self synthesis tool to decide their realizability.
If all the sub-specifications are realizable, the systems obtained are
then composed into a single system for the original specification,
which also shares the structure of the mode decomposition.


\section{Preliminaries}
\label{sec:preliminaries}
We consider a finite set of $\AP$ of atomic propositions.
Since we are interested in reactive systems where there is an ongoing
interaction between a system and its environment, we split $\AP$ into
those propositions controlled by the environment $\VX$ and those
controlled by the system $\VY$, so $\VX\cup\VY=\AP$ and
$\VX\cap\VY=\emptyset$.
The alphabet induced by the atomic propositions is $\Sigma=2^\AP$.
We use $\Sigma^*$ for the set of finite words over $\Sigma$ and
$\Sigma^\omega$ for the set of infinite words over $\Sigma$.
Given $\sigma\in \Sigma^\omega$ and $i\in\Nat$, $\sigma(i)$ represents
the element of $\sigma$ at position $i$, and $\sigma^i$ represents the
word $\sigma'$ that results by removing the prefix
$\sigma(0)\ldots\sigma(i-1)$ from $\sigma$, that is $\sigma'$
s.t. $\sigma'(j)=\sigma(j-1)$ for $j\geq i$.
Given $u\in \Sigma^*$ and $v\in \Sigma^\omega$, $uv$ represents the
$\omega$-word that results from concatenating $u$ and $v$.
We use LTL~\cite{pnueli77temporal,manna95temporal} to describe
\changing{specifications. The syntax of LTL is the following:}
\[
  \varphi  ::= \true \DefOR a \DefOR \varphi \lor \varphi \DefOR \neg \varphi
  \DefOR \Next \varphi \DefOR \varphi \U \varphi \DefOR \G \varphi
\]
\changing{where $a\in \AP$, and $\lor$, $\land$ and $\neg$ are the usual Boolean
  disjunction, conjunction and negation, and $\Next$ is the next temporal operator
  (a common derived operator is \false = $\Not \true$).
}
%
%
%
A formula with no temporal operator is called a Boolean formula, or
predicate.
%
%
%
We say $\varphi$ is in negation normal form (\NNF), whenever all
negation operators in $\varphi$ are pushed only in front of atoms
using dualities.
The semantics of \LTL associate traces $\sigma\in\Sigma^\omega$ with
formulae as follows:
 \[
   \begin{array}{l@{\hspace{1em}}l@{\hspace{0.3em}}l}
     \sigma \models \true && \text{always holds} \\
     \sigma \models a & \text{iff } & a \in\sigma(0) \\
     \sigma \models \varphi_1 \Or \varphi_2 & \text{iff } & \sigma\models \varphi_1 \text{ or } \sigma\models \varphi_2 \\
     \sigma \models \neg \varphi & \text{iff } & \sigma \not\models\varphi \\
     \sigma \models \Next \varphi & \text{iff } & \sigma^1\models \varphi \\
      \sigma \models \varphi_1 \U \varphi_2 & \text{iff } & \text{for some } i\geq 0\;\; \sigma^i\models \varphi_2, \text{ and } \text{for all } 0\leq j<i, \sigma^j\models\varphi_1 \\
     \sigma \models \G \varphi & \text{iff } & \text{for all } i\geq 0\;\; \sigma^i\models \varphi\\
   \end{array}
 \]

 \subsubsection{A Syntactic Fragment for Safety.}
\changing{A useful fragment of \LTL is \LTLX where formulas only
  contain $\X$ as a temporal operator.
In this work, we focus on a fragment of \LTL we called \GXz:
\[
  \alpha \Into (\beta \And \Always \psi)
\]
where $\alpha$, $\beta$ and $\psi$ are in \LTLX.

This fragment can only express safety
properties~\cite{manna95temporal,chang05characterization} and includes
a large fragment of all safety properties expressible in LTL.
This format is supported by tools like Strix~\cite{meyer2018strix}
and is convenient for our reactive problem specification.}

%
%
%
%
%
%
 
\changing{
\begin{definition}[Reactive Specification]
  \label{def:react}
  A reactive specification $\contractfull$ is given by
  $A=(\IniE,\AssE)$ and $G=(\IniS,\AssS)$ (all $\LTLX$ formulas),
  where $\IniE$ and $\IniS$ are the initial conditions of the
  environment and the system, and $\AssE$ and $\AssS$ are called
  assumptions and guarantees. The meaning of $S$ is the $\GXz$
  formula:
  \[
    (\IniE \rightarrow (\IniS \land \G(\AssE \rightarrow\AssS)))
  \]
\end{definition}

\noindent In TLSF $\IniE$ and $\IniS$ are represented as
\emph{INITIALLY} and \emph{PRESET}, resp.}

 %
 %
 %
%
%
%
%
%


\subsubsection{Reactive Synthesis.}
Consider a specification $\varphi$ over $\AP = \VX \cup \VY$.
A trace $\sigma$ is formed by the environment and the system choosing
in turn valuations for their propositions.
The specification $\varphi$ is realizable with respect to $(\VX, \VY)$
if there exists a strategy
$g: (2^\VX)^{+} \rightarrow 2^\VY$ such that for
an arbitrary infinite sequence
$X = X_0,X_1,X_2,\ldots \in (2^\VX)^{\omega}$, $\varphi$ is
\true{} in the infinite trace
$\rho = (X_0 \cup g(X_0)),(X_1 \cup g(X_0,X_1)),(X_2 \cup
g(X_0,X_1,X_2)),\ldots$
A play $\rho$ is \textit{winning} (for the system) if
$\rho \models \varphi$.

Realizability is the decision problem of whether a specification has a
winning strategy, and synthesis is the problem of computing one wining
system (strategy).
Both problems can be solved in double-exponential time for an
arbitrary LTL formula~\cite{PnueliRosner1989}.
If there is no winning strategy for the system, the specification is
called \textit{unrealizable}.
In this scenario, the environment has at least one strategy to falsify
$\varphi$ for every possible strategy of the system.
Reactive safety synthesis considers reactive synthesis for safety formulas.

We encode system strategies using a deterministic Mealy machine
$W=(Q,s,\delta,L)$ where $Q$ is the set of states, $s$ is the initial
state, $\delta:Q\times{}2^\VX\Into{}Q$ is the transition function that
given valuations of the environment variables it produces a successor
state and $L:Q\times{}2^\VX\Into{}2^\VY$ is the output labeling that
given valuations of the environment it produces valuations of the
system.
\changing{The strategy $g$ encoded by a machine
$W:(Q,s,\delta,L)$ 
is as follows:}
\begin{itemize}
\item if $e\in 2^\VX$, then $g(e) = L(s,e)$
\item if $u\in (2^\VX)^+$ and $e\in 2^\VX$ then $g(ue) = L(\delta^*(s,u),e)$ where
  $\delta^*$ is the usual extension of $\delta$ to $(2^\VX)^*$.
\end{itemize}
It is well known that if a specification is realizable then there is
Mealy machine encoding a winning strategy for the system.

\section{Mode Based Synthesis}
\label{sec:solution}

We present now our mode-based solution to reactive safety synthesis.
The starting point is a \emph{reactive specification} as a \GXz
formula written in TLSF.
%
\changing{We define a mode $m$ as a predicate over $\VX\cup\VY$, that is $m \in 2^{\VX\cup\VY}$}.
A mode captures a set of states of the system during its execution.
Given a trace $\sigma=s_0,s_1,\ldots$, if $s_i\models m$ we say that
$m$ is the \emph{active mode} at time $i$.
In this paper, we consider mutually exclusive modes, so only one mode
can be active at a given point in time.
As part of the specification of synthesis problems the requirement
engineer describes the modes $M = \{m_1,\ldots,m_n\}$, partially
expressing the intentions of the structure of the intended system.
A set of modes $M = \{m_1, m_2, \ldots, m_n\}$ is legal if it
partitions the set of variable valuations, that is:
\begin{compactitem}
\item \textbf{Disjointness}:
  for all $i\neq j$,  $(m_i \Impl \Not m_j)$ is valid.
\item \textbf{Completeness}: $\bigvee_{i} m_i$ is valid.
\end{compactitem}

\changing{Within a trace $\sigma$ there may be instants during
  execution there are transitions between modes.  We will refer to the
  modes involved in this transition as \emph{related modes}. Formally:
\begin{definition}[Related Modes]
  Consider a trace $\sigma = \sigma(0)\sigma(1)\sigma(2)\ldots$ and
  two modes $m_1, m_2 \in M$.  We say that $m_1$ and $m_2$ as related,
  denoted as $m_1 \prec m_2$ if, at some point $i$:
  $(\sigma(i) \models m_1)$ and $(\sigma({i+1}) \models m_2)$.
\end{definition}
}

A key element of our approach is to enrich the specification of the
synthesis sub-problem corresponding to mode $m_i$ forbidding the
system to jump to another mode $m_j$ unless the initial condition of
mode $j$ satisfying the pending ``obligations'' at the time of
jumping.
To formally capture obligations we introduce fresh variables for
future sub-formulas that appear in the specification.
\begin{definition}[\textit{Obligation Variables}]
  \label{def:ov}
  For each sub-formula $\Next\psi$ in the specification, we introduce
  a fresh variables $\sv{\Next\psi}$ to encodes that the system is
  obliged to satisfy $\psi$.
\end{definition}
These variables will be controlled by the system and their dynamics
will be captured by $s_{\Next\psi} \Into \Next\psi$
%
introduced in every mode (unless the system leaves the mode, which
will be allowed only if the arriving system satisfy $\psi$).
%
%
\changing{These variables are similar to temporal
  testers~\cite{pnueli06psl} and allow a simple treatment of
  obligations that are left pending after a mode jump}.
We also introduce variables $\jump_j$ which will encode (in the game
and sub-specification corresponding to mode $m_i$) whether the system
decides to jump to mode $m_j$ (see Alg.~\ref{alg:projection}
below).
\subsection{Mode Based Decomposition}
\label{sec:lvp}

We present now the \mobi algorithm, which decomposes a
\changing{reactive specification} $\contract$ into a set of (smaller)
specifications $\Pi = \{\contract_1,\ldots,\contract_n\}$, using the
provided system modes $M = \{m_1,\ldots,m_n\}$ and initial
mode-conditions $I = \{I_1,\ldots, I_n\}$.
Fig.~\ref{fig:approach} shows an overview of \mobi.
\changing{Particularly, \mobi receives a specification together with
  modes and one initial condition per mode.  The algorithm decompose
  the specification into smaller sub-specifications one per mode.}
%
%
\begin{figure}[t!]
\centering
\includegraphics[width=\textwidth]{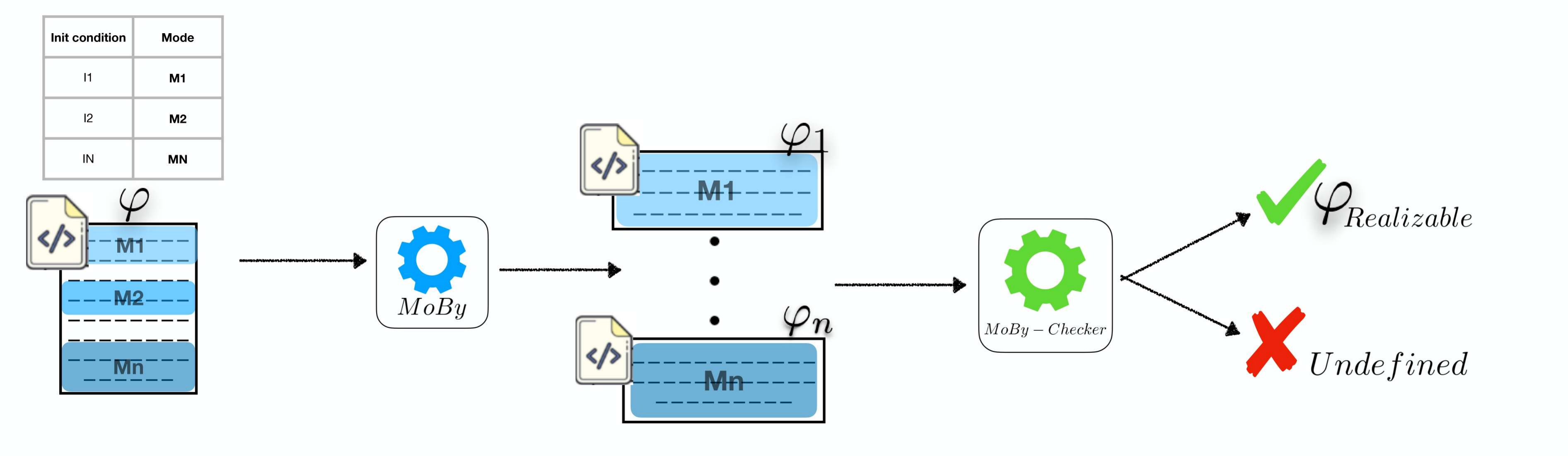}
\caption{Overview of \mobi}
\label{fig:approach}
\end{figure}
\newcommand{\AlgoRemoveModes}{
  \begin{algorithm}[H]
    \centering
    \footnotesize
    \begin{algorithmic}[1]
      \Function{\RmModes}{$\varphi,m$}
      \For {each $f \in \ASF(\varphi)$}
      \If {$(m \Impl f)$ is valid}
      \State $\varphi \leftarrow \varphi[f \backslash \TRUE]$
      \EndIf
      \If{$(m \Impl\Not{}f)$ is valid}
      \State $\varphi \leftarrow \varphi[f \backslash \FALSE]$
      \EndIf
      \EndFor
      \State \textbf{return} $\Call{\Simpl}{\varphi}$
      \EndFunction
    \end{algorithmic}
    \caption{Simplify (remove)}
    \label{alg:auxf}
  \end{algorithm}
}
\newcommand{\AlgoAddJumps}{
        \begin{algorithm}[H]
            \centering
            \footnotesize
            \caption{Add jumps}
            \label{alg:auxf2}
            \begin{algorithmic}[1]
                \Function{add\_jmp}{$m, \varphi$}
                \State $B=\{f\in\NSF(\varphi)\;|\;m\Impl\neg f \text{ is valid}\}$
                \For{$g\in B$}
                \State $\varphi\leftarrow\varphi[\X f \backslash S_{\X m}]$
                \EndFor
                \State \textbf{return} $\varphi$
                \EndFunction
            \end{algorithmic}
          \end{algorithm}
}

The main result is that the decomposition that \mobi performs
guarantees that if each projection $S_i \in \Pi$ is realizable then
the original specification is also realizable, and that the systems
synthesized independently for each sub-specification can be combined
into an implementation for the original specification $\contract$
\changing{(See Lemma~\ref{lemma:moby2} and Corollary~\ref{cor:mobi})}.

We first introduce some useful notation before presenting the main
algorithm.
We denote by $\varphi[\phi \backslash \psi] $ the formula that is
obtained by replacing in $\varphi$ occurrences of $\phi$ by $\psi$.
We assume that all formulas have been converted to \NNF, where $\Next$
operators have been pushed to the atoms.
It is easy to see that a formula in \NNF is a Boolean combination of
sub-formulas of the form $\Next^i p$ where $p\in \AP$ and sub-formulas
$\psi$ that do not contain any temporal operator.
We use some auxiliary functions:
\begin{compactitem}
\item The first function is $\ASF(\varphi)$, which returns the set of
  sub-formulas $\psi$ of $\varphi$ such that (1) $\psi$ does not contain
  $\X$ (2) $\psi$ is either $\varphi$ or the father formula of
  $\psi$ contains $\Next$.
  We call these formulas maximal next-free sub-formulas of $\varphi$.
\item The second function is $\NSF(\varphi)$, which returns the set of
  sub-formulas $\psi$ such that (1) the root symbol of $\psi$ is $\X$
  and (2) either $\psi$ is $\varphi$, or the father of $\psi$ does not
  start with $\Next$.
  It is easy to see that all formulas returned by $\NSF$ are of the
  form $\Next^i p$ for $i>0$, and indeed are the sub-formulas of the form
  $\Next^i p$ not contain in other formulas other sub-formulas of these forms.
  We call these formulas the maximal next sub-formulas of $\varphi$.
\end{compactitem}
For example, let $\varphi = \X p \rightarrow (\X q \land r)$, which is
in \NNF.
$\ASF(\varphi)=\{r\}$, as $r$ is the only formula that does not
contain $\X$ but its father formula does.
$\NSF(\varphi)=\{\X p, \X q\}$.
We also use the following auxiliary functions:
\begin{compactitem}
\item $\Simpl(\varphi)$, which performs simple Boolean
  simplifications, including $\true \And \varphi \mapsto \varphi$,
  $\false \And \varphi \mapsto \false$,
  $\true \Or \varphi \mapsto \true$,
  $\false \Or \varphi \mapsto \varphi$, etc.
\item $\RmNext$, which takes a formula of the form $\X^{i} \varphi$
  and returns $\X^{i-1} \varphi$.
\item $\Fresh$, which takes a formula of the form $\X^{i} \varphi$ and
  returns the obligation variable $s_{\X^{i} \varphi}$. This function
  also accepts a proposition $p\in \AP$ in which case it returns $p$
  itself.
\end{compactitem}
The output of $\Simpl(\varphi)$ is either $\true$ or $\false$, or a
formula that does not contain $\true$ or $\false$ at all.
The simplification performed by $\Simpl$ is particularly useful

\begin{wrapfigure}[9]{l}{0.42\textwidth}
\begin{minipage}[t]{0.42\textwidth}
    \vspace{-4.5em}
    \AlgoRemoveModes
  \end{minipage}
\end{wrapfigure}
\noindent%
simplifying $(\false \Into \psi)$ to $\true$, because given a
requirement of the form $C\Into{}D$, if $C$ is simplified to $\False$
in a given mode then $C\Into D$ will be simplified to $\true$ ignoring
all sub-formulas within $D$.
We introduce $\RmModes(\varphi,m)$ on the left, which given a mode
$m$ and a formula $\varphi$ simplifies $\varphi$ under the assumption
that the current state satisfies $m$, that is, specializes $\varphi$
for mode $m$.\\
\begin{example}
  Consider $m_1 :(\Counter_1\And\Not\Counter_2)$, and
  $\varphi_1 :\lnot \Counter_2 \rightarrow \lnot \Trigger$ and
  $\varphi_2 : (\Counter_1 \land \lnot \RESET) \rightarrow \X
  (\Counter_2 \lor \RESET)$.
Then, 
\begin{align*}
  \RmModes(\varphi_1,m_1)&=\lnot \Trigger\\
  \RmModes(\varphi_2,m_1)&=\lnot \RESET \rightarrow \X (\Counter_2 \lor \RESET)
\end{align*}
Finally, $\Fresh(\X\psi)=s_{\X\psi}$.
\end{example}

\subsection{The Mode-Base Projection Algorithm \mobi}
\newcommand{\AlgoMobi}{
  \begin{algorithm}[b!] 
    \caption{\mobi: Mode-Based Projections.}
    \begin{algorithmic}[1]
      \State \textbf{Inputs:} $S:(A,G), M : \{m_1,\ldots,m_n\}, I :\{I_1,\ldots,I_n\}$.
      \State \textbf{Outputs:} $\PR = [\Pi_1,\ldots,\Pi_n]$.
      \Function{\ComputeProjections}{$S,M,I$}
      \For {each mode index $i\in\{1\ldots{}n\}$}
      \State $G'\leftarrow \Reduce(G,m_i)$
      \State $\Oblig \leftarrow \changing{\NSF(G')}$ 
      \State $G_i=\{ \neg\done\Then m_i \}$
      \For {each requirement $\psi \in G'$}
     \State $\psi'\leftarrow$ replace $f$ for $\Fresh(f)$ in $\psi$ (for all $f\in \Oblig$)
      \State $G_i.\Add((\neg\done)\Into\psi')$
      \EndFor
      \For {each obligation subformula $f\in\Oblig$} 
      \State $G_i.\Add((\neg\done \And \Fresh(f))\Into\Next\Fresh(\RmNext(f)))$ 
      \EndFor
      \For {each mode $j\neq i$ such that $m_i \prec m_j$ and for every $f\in\Oblig$} 
       \If{$(I_j\Into \RmNext(f))$ is not valid}
       \State $G_i.\Add(\jump_j\Then \neg \Fresh(f))$
       \EndIf
       \EndFor
       \State $G_i.\Add(\done \Then \Next\done)$
       \State $G_i.\Add(\phantom{\Not}(\,\bigvee_j \jump_j) \Then \Next\done)$
       \State $G_i.\Add((\Not\bigvee_j \jump_j)\Then(\Not\done\Then\Next\Not\done))$
       \State $G_i.\Add(\bigwedge_{j\neq k} \jump_j \Then \neg\jump_k)$
      \State $\PR[i] \leftarrow (A,G_i)$
      \EndFor
      \State \textbf{return} $\PR$
      \EndFunction
      \Function{\Reduce}{$\Phi,m$}
      \State \textbf{return} $\{ \RmModes(\varphi,m)  \;|\; \varphi\in\Phi \}$
      \EndFunction
    \end{algorithmic}
    \label{alg:projection}
  \end{algorithm}
}

\newcommand{\AlgoMobiOld}{
\begin{algorithm}[t!] 
    \caption{\mobi: Mode-Based Projections.}
    \begin{algorithmic}[1]
        \State \textbf{Inputs:} $\contractfull, \mathcal{M} = \{m_1,\ldots,m_k\}, \mathcal{I} = \{i_1,\ldots,i_n\}$.
        \State \textbf{Outputs:} $Pr = [\Pi_1,\ldots,\Pi_k]$.
        \Function{increasing\_phase}{$Pr, \mathcal{M}$}
            \For {each index mode $(i,j) i \neq j \in \mathcal{M}$}
                \If{$\mathcal{M}.get(i) \not\prec \mathcal{M}.get(j)$} \Comment{Modes are not related}
                    \State \textbf{continue}
                \EndIf
                \State \textit{curr\_mode} $\leftarrow Pr[i]$ 
                \State $fresh\_cmode \leftarrow \emptyset$ \Comment{Set of fresh variables for current mode}
                \For {each requirement $\psi \rightarrow \varphi \in$ \textit{curr\_mode}}
                    \State $G' \leftarrow (\psi \rightarrow \varphi)$                            
                    \State $\varphi' \leftarrow \Call{rm\_modes}{curr\_mode,\varphi}$
                    \If{$\Call{deep}{\varphi} == 0 \lor \varphi' \not= \FALSE$} \Comment{same mode or not related to any}
                        \State \textbf{continue}    
                    \EndIf
                    \While{$\lnot finished$}
                        \State $v, ant \leftarrow \Call{fresh}{\varphi}$
                        \State $fresh\_cmode.add(v)$
                        \If{$\Call{deep}{\varphi} == 1$}
                                \State $G'[\varphi \backslash v]$
                                \State $G'.add(ant \land \lnot jump_j \rightarrow \Call{rm\_modes}{curr\_mode,\varphi})$
                                \State $finished \leftarrow \TRUE$
                        \Else
                            \State $G'[\varphi \backslash \X(v)]$
                            \State $\varphi \leftarrow \Call{rm\_nx}{\varphi}$
                            \State $G'.add(ant \land \lnot jump_j \rightarrow \varphi)$
                        \EndIf
                    \EndWhile
                \EndFor
                \State $fresh\_cmode \leftarrow filter (\lambda x . (I[j] \rightarrow \lfloor x \rfloor$ is not valid) ($fresh\_cmode$)
                \State $G'.add(jump_j \rightarrow ((\bigwedge\limits_{i=0}^{size(fresh\_cmode)} \lnot fresh\_cmode[i]) \land jump))$
                \State $Pr[i].G \leftarrow G' \cup (\G jump \rightarrow \X \G jump)$
            \EndFor
        \EndFunction
        \Function{reduction\_phase}{$\varphi, mode$}
            \State $\varphi' \leftarrow \emptyset$
            \For {each requirement $r \in \varphi$}
                \State $r \leftarrow \Call{rm\_Modes}{mode,r}$
                \State $\varphi' \leftarrow \varphi' \cup \{r\}$
            \EndFor
            \State \textbf{return} $\varphi'$ 
        \EndFunction
        \Function{MoBy}{}
            \State $Pr \leftarrow \emptyset$
            \For {each mode $m \in \mathcal{M}$}
                \State $G' \leftarrow \emptyset$
                \For {each requirement $r \in G$}
                    \State $r \leftarrow \Call{reduction\_phase}{m,r}$
                    \State $G' \leftarrow G' \cup \{\G \lnot jump \rightarrow r\}$
                \EndFor
                \State $Pr \leftarrow Pr.append(\contract_m = (A,G'))$
            \EndFor
            \State \Call{increasing\_phase}{$Pr,\mathcal{M}$}
            \State \textbf{return} $Pr$
        \EndFunction
    \end{algorithmic}
    \label{alg:projections}
\end{algorithm}
}
\changing{As mentioned before our algorithm takes as a input a reactive specification $\contract$}
an indexed set $M = \{m_1,\ldots, m_n\}$ of modes and an indexed set 
$I = \{I_1,\ldots,I_n\}$ of initial conditions, one for each mode.
We first add to each $I_i$ the predicate $\Not\done$, to encode that
in its initial state a sub-system that solves the game for mode $m_i$
has not jumped to another mode yet.
For each mode $m_i$, \mobi specializes all guarantee formulas calling
$\RmModes$, and then adds additional requirements for the obligation
variables and to control when the system can exit the mode.
Alg.~\ref{alg:projection} presents \mobi in pseudo-code.

\AlgoMobi

Line $5$ simplifies all requirements specifically for mode $m_i$, that
is, it will only focus on solving all requirements for states that
satisfy $m_i$.
Line $7$ starts the goals for mode $i$ establishing that unless the
system has jumped to another mode, the mode predicate $m_i$ must hold
in mode $i$.
Lines $8$ to $10$ substitute all temporal formulas in the requirements
with their obligation variables, establishing that all requirements
must hold unless the system has left the mode.
Lines $11$ to $12$ establish the semantics of obligation variables,
forcing their temporal behavior as long as the system stays within the
mode ($\neg\done$).
Lines $13$ to $15$ precludes the system to jump to another mode $m_j$
if $m_j$ cannot fulfill pending promises.
Lines $16$ to $18$ establish that once the system has jumped the game
is considered finished, and that the system is only finished jumping
to some other mode.
Finally, line $19$ limits to jump to at most one mode.

\begin{example}
  We apply \mobi to the example in Fig.~\ref{lst:counter} for $N=2$,
  with three modes
  $M = \{m_1:\{counter_0\}, m_2:\{counter_1\}, m_3:\{counter_2\}\}$.
  The initial conditions only establish the variable of the mode is
  satisfied $I_1=m_1$, $I_2=m_2$, $I_3=m_3$ (only forcing $\neg\done$
  as well).
  The \mobi algorithm computes the following  projections:

\begin{lstlisting}[language=TLSF, mathescape=true]
INPUTS reset; start;
ASSUMPTIONS G !(reset && start); INITIALLY (!reset && !start) || reset
[Projection_1]                                                        
OUTPUTS counter_0; trigger; $s_{\X\varphi}$; $\emph{jump}_2$; $\emph{done}$
GUARANTEES                                                             
G (!$\emph{done}$ -> (counter_0))                                                    
G (!$\emph{done}$ -> (reset -> X counter_0));                                        
G (!$\emph{done}$ -> (start -> $s_{\X\varphi}$));                                     
G (!$\emph{done}$ -> (($s_{\X\varphi}$ && !$\emph{done}$) -> X FALSE));               
G (!$\emph{done}$ -> (!trigger));                                                     
G ($\emph{done}$ -> X $\emph{done}$);                                                 
G ($\emph{jump}_2$ -> X $\emph{done}$);                                               
G (!$\emph{jump}_2$ -> (!$\emph{done}$ ->  X !$\emph{done}$));                        
                                                                                     
[Projection_2]                                                  
OUTPUTS counter_1; trigger; $\emph{jump}_1$, $\emph{jump}_3$ $s_{\X\varphi}$; $s_{\X\varphi_1}$;                                       
GUARANTEES                                                      
G !$\emph{done}$ -> (counter_1)                                 
G !$\emph{done}$ -> (reset -> $s_{\X\varphi}$);                         
G !$\emph{done}$ -> ($s_{\X\varphi}$ && !$\emph{done}$ -> X FALSE);                    
G !$\emph{done}$ -> (!reset -> $s_{\X\varphi_1}$);                        
G !$\emph{done}$ -> ($s_{\X\varphi_1}$ && !$\emph{done}$ -> X FALSE);                 
G !$\emph{done}$ -> (!trigger);                               
G (($s_{\X\varphi}$ || $s_{\X\varphi_1}$) -> X $\emph{done}$);                        
G $\emph{jump}_1$ -> !$s_{\X\varphi_1}$;
G $\emph{jump}_3$ -> !$s_{\X\varphi}$;
G (!($s_{\X\varphi}$ || $s_{\X\varphi_1}$)  -> (!$\emph{done}$ ->  X !$\emph{done}$));

[Projection_3]
OUTPUTS counter_2; trigger; $\emph{jump}_1$; $s_{\X\varphi}$ $\emph{jump}_1$
GUARANTEES
G (!$\emph{done}$ -> (counter_2))
G (!$\emph{done}$ -> (reset -> $s_{\X\varphi}$))
G (!$\emph{done}$ -> ($s_{\X\varphi}$ && !$\emph{done}$ -> X FALSE));
G (!$\emph{done}$ -> (counter_2 -> $s_{\X\varphi}$));
G (!$\emph{done}$ -> (trigger)); 
G ($\emph{jump}_1$ -> X $\emph{done}$);
G (!$\emph{jump}_1$ -> (!$\emph{done}$ ->  X !$\emph{done}$));
\end{lstlisting}
\end{example}      
    

%
%
\subsection{Composing solutions}
After decomposing $\contract$ into a set of projections
$\PR = \{\Pi_1,\ldots, \Pi_n\}$ using \mobi, Alg.~\ref{alg:moby2}
composes winning strategies for the system obtained for each mode into
a single winning strategy for the original specification $\contract$.
\begin{algorithm}[b!]
\caption{Composition of Winning Strategies}
\label{alg:moby2}
\begin{algorithmic}[1]
  \State \textbf{Input:} A winning strategy
  $W_i = (Q_i,s_i, \delta_i,L_i)$ for each projection $p_i\in Pr$.
  \State \textbf{Output:} A composed winning strategy $W = (Q,s,\delta,L)$.
  \Function{\Compose}{$W_1,\ldots,W_n$}
  \State $Q \leftarrow \bigcup_{i=1}^n Q_i$
  \State $s \leftarrow s_1$
  \State $\delta \leftarrow \emptyset$
  \For {each mode index $i \in \{1\ldots{}n\}$}
  \State $(Q_i,s_i,\delta_i,L_i)\leftarrow W_i$
\For {each $(q, a)\in Q_i\times 2^\VX$}
\State $L(q,a)\leftarrow{}L_i(q,a)$ 
\If {$\delta_i(s,a)\models\jump_j$ for some $j$}
\State $\delta(q,a)\leftarrow{}s_j$
\Else  \Comment{$\jump_j\notin\delta_i(q,a)$ for any $j$}
\State $\delta(q,a)\leftarrow\delta_i(q,a)$
\EndIf
\EndFor
\EndFor
\State \textbf{return} $W:(Q,s,\delta,L)$
      \EndFunction
\end{algorithmic}
\end{algorithm}

\begin{lemma}[Composition's correctness]
\label{lemma:moby2}
Let $M = \{m_1, \ldots, m_n\}$ and $I = \{I_1, \ldots, I_n\}$ be a set
of valid mode descriptions for a specification $\contract$, and let
$St = \{W_1, \ldots, W_n\}$ be a set of winning strategies for each
projection $p \in Pr = \{\Pi_1, \ldots, \Pi_n\}$ Then, the composed
winning strategy $W$ obtained using Alg.~\ref{alg:moby2} is a winning
strategy for $\contract$.
\end{lemma}

\begin{proof}
  Let $\contract$ be a specification, $M = \{m_1, m_2, \ldots, m_n\}$ and
  $I=\{I_1,\ldots,I_n\}$ a mode description.
  Also, let's consider $Pr = \{\Pi_1, \ldots, \Pi_n\}$ be the projection generated by
  Alg.~\ref{alg:projection}.
  We assume that all sub-specifications are realizable.
  Let $St = \{W_1, \ldots, W_n\}$ be winning strategies for each
  of the sub-specifications and let $W:(Q,s,\delta,L)$ be the strategy
  for the original specifications generated by Alg.~\ref{alg:moby2}.
  We will show now that $W$ is a winning strategy.
  The essence of the proof is to show that if a mode $m_j$ starts at
  position $i$ and the system follows $W$, this corresponds to follow
  $W_j$.
  In turn, this guarantees that $\PR[j]$ holds until the next mode is
  entered (or ad infinitum if no mode change happens), which
  guarantees that $\contract$ holds within the segment after the new
  mode enters in its initial state.
  By induction, the result follows.

  By contradiction, assume that $W$ is not winning and let
  $\rho\in 2^{\VX\cup\VY}$ be a play that is played according to $W$
  that is loosing for the system.
  In other words, there is position $i$ such that $\rho^i$ violates
  some requirement in $\contract$.
  Let $i$ be the first such position.
  Let $m_j$ be the mode at position $i$ and let $i'<i$ be the position
  at which $m_j$ is the mode at position $i'$ and either $i'=0$ or the
  mode at position $i'-1$ is not $m_j$.
  \begin{compactitem}
    \item If $i'=0$, between $0$ and $i$, $W$ coincides with $W_j$.
      Therefore, since $W_j$ is winning $\Pi[j]$ must hold at $i$, which
      implies that $\contract$ holds at $i$, which is contradiction.
    \item Consider now the case where $i'-1$ is not $m_j$, but some
      other mode $m_l$.
      Then, since in $m_l$ is winning $W_l$, it holds that $\PR[l]$
      holds at $i'-1$ so, in particular all pending obligations are
      implied by $I_j$.
      Therefore, the suffix trace $\rho^{i'}$ is winning for $W_j$.
      Again, it follows that $\contract$ holds at $i$, which is a
      contradiction.
    \end{compactitem}
    Hence, the lemma holds.
    \qed
\end{proof}
The following corollary follows immediately.

\begin{corollary}[Semi-Realizability]
  \label{cor:mobi}
  Given a specification $\contract$, a set $M$ of valid system
  modes and a set $I$ of initial conditions. If all
  projections generated by \mobi are realizable, then $\contract$ is
  also realizable.
\end{corollary}

\newcommand{\RQone}{\textcolor{blue}{RQ1}\xspace}
\newcommand{\RQtwo}{\textcolor{blue}{RQ2}\xspace}
\newcommand{\RQthree}{\textcolor{blue}{RQ3}\xspace}

\section{Empirical Evaluation}
\label{sec:empirical}
\label{sec:research-questions}

We implemented \mobi in the \textit{Java} programming language using
the well-known \textit{Owl} library~\cite{Kretinsky+2018} to
manipulate LTL specifications.
\mobi integrates the LTL satisfiability checker
Polsat~\cite{DBLP:journals/corr/LiP0YVH13}, a portfolio consisting of
four LTL solvers that run in parallel.
%
%
To perform all realizability checks, we discharge each
sub-specification to
\textit{Strix}~\cite{meyer2018strix}.
All experiments in this section were run on a cluster equipped with a
Xeon processor with a clock speed of 2.6GHz, 16GB of RAM, and running
the GNU/Linux operating system.

We report in this section an empirical evaluation of \mobi.
We aim to empirically evaluate the following research questions:
\begin{compactitem}
\item \RQone: \emph{How effective is \mobi in decomposing mode-based specifications?}
\item \RQtwo: \emph{Does \mobi complement state of the art synthesis tools?}
\item \RQthree: \emph{Can \mobi be used to improve the synthesis time?}
\end{compactitem}

\begin{figure}[b!]
  \centering
  \small
\begin{tabular}{|c|r|r|r|r|}
\hline
Case	& \#A - \#G & \#Modes & \#In & \#Out\\
\hline
10-Counter-Machine & 2-15 & [2,5,10]  & 2 & 12\\
20-Counter-Machine & 2-25 & [2,5,10,20]  & 2 & 22\\
50-Counter-Machine & 2-55 & [2,5,10,50]  & 2 & 52\\
100-Counter-Machine & 2-105 & [2,5,10,50,100]  & 2 & 102\\
Minepump & 3-4 & [2] & 300 & 5\\
Sis(n) & 2-7 & [3]  & (2+n) & 7\\
Thermostat(n) & 3-4 & [3]  & (31+n) & 4\\
Cruise(n) & 3-15 & [4]  & (5+n) & 8\\
AltLayer(n) & 1-9 & [3]  & n & 5\\
Lift(n) & 1-187 & [3]  & n & (4+n)\\
\hline
\end{tabular}
\caption{Assumptions (A), Guarantees (G), Modes, Variables}
\label{tab:case-studies}
\end{figure}

\begin{figure}[t!]
    \centering
    \small
    \begin{tabular}{|l|r||r|r||r|r|r|r||}\hline
      Specification & \#Modes & \multicolumn{2}{|c||}{Synthesis Time (s)} & \multicolumn{4}{|c||}{Specification Size} \\
      \cline{3-8}
      & & \multirow{2}{*}{Monolithic} & \multirow{2}{*}{\mobi} & \multicolumn{2}{|c|}{Monolithic} & \multicolumn{2}{|c||}{\mobi}\\
      \cline{5-8}
      & & & & \#Clauses & Length & \#Clauses & Length \\
      \hline \hline
      \multirow{3}{*}{\textit{CM10}} 
      & 2  & \multirow{3}{*}{26} & 0.32 & \multirow{3}{*}{48}  & \multirow{3}{*}{252} & 28 & 117 \\  
      & 5  &                      & 0.67 &                     &  & 8 & 30 \\  
      & 10  &                     & 0.58 &                     &  & 8 & 39 \\ 
      \hline
      \multirow{3}{*}{\textit{CM20}} 
      & 2  & \multirow{4}{*}{\textbf{Timeout}} & 3.62 & \multirow{4}{*}{88}  & \multirow{4}{*}{672} & 48 & 252 \\  
      & 5  &                      & 1.15 &                     &  & 24 & 96 \\  
      & 10  &                     & 2.06 &                     &  & 16 & 60 \\    
      & 20  &                     & 1.08 &                  &  & 8 & 49 \\    
      \hline 
      \multirow{3}{*}{\textit{CM50}} 
      & 2  & \multirow{4}{*}{\textbf{Timeout}} & 2.56 & \multirow{4}{*}{208}  & \multirow{4}{*}{3132} & 136 & 1036  \\  
      & 5  &                      & 19 &                     &      &  48 & 256 \\  
      & 10  &                     & 3 &                     &       &  28 &  117 \\    
      & 50  &                     & 1.67 &                     & &   8 & 79 \\    
      \hline 
      \multirow{5}{*}{\textit{CM100}} 
      & 2  & \multirow{4}{*}{\textbf{Timeout}} & \textbf{Timeout} & \multirow{4}{*}{408}  & \multirow{4}{*}{11232} & 208 &  3132\\  
      & 5  &                      & 5.12 &                     &  &  88 &  672\\  
      & 10  &                     & 4 &                     &  &  48 &  252\\    
      & 50  &                     & 9 &                     &  &  19 &  57 \\    
      & 100  &                    & 3.23 &                  &  &   8 & 129 \\   
      \hline
      \textit{Minepump} & 2  & 140 & 90 & 11598  & 21365 & 5800 & 10685 \\
      \hline 
      \textit{Sis-250} & 3  & 18                  & 2 & 521  & 1072 & 133 & 287 \\  
      \hline 
      \textit{Sis-500} & 3  & 96                & 4 & 1021  & 2072 & 258 & 538 \\  
      \hline
      \textit{Sis-1000}  & 3  & \textbf{Timeout}  & 11 & 2021  & 4072 & 508 & 1035 \\  
      \hline
      \textit{Sis-1500} & 3  & \textbf{Timeout} & 20 & 3021  & 6072 & 758 & 1538 \\  
      \hline
      \textit{Sis-2000} & 3  & \textbf{Timeout} & 38 & 4021  & 8072 & 1258 & 2300 \\  
      \hline
      \textit{Sis-4000} & 3  & \textbf{Timeout} & 157 & 8021  & 16072 & 2678 & 3560 \\  
      \hline
      \textit{Sis-4500} & 3  & \textbf{Timeout} & 172 & 9020  & 18040 & 3006 & 4002 \\  
      \hline
      \textit{Sis-5000} & 3  & \textbf{Timeout} & 268 & 10020  & 20040 & 3340 & 4447 \\    
      \hline
      \textit{Thermostat-10} & 3  & 1 & 1 & 73  & 151 & 42 & 97 \\    
      \hline
      \textit{Thermostat-20} & 3  & \textbf{Timeout} & 1 & 172  & 276 & 75 & 152 \\    
      \hline
      \textit{Thermostat-100} & 3  & \textbf{Timeout} & 10 & 4032  & 4416 & 1375 & 1652 \\    
      \hline
      \textit{Thermostat-200} & 3  & \textbf{Timeout} & 48 & 12132  & 12916 & 4075 & 4619 \\ 
      \hline
      \textit{Cruise-150} & 4  & 75 & 63 & 15339  & 15855 & 6824 & 7067 \\ 
      \hline
      \textit{Cruise-200} & 4  & 132 & 100 & 30039  & 30756 & 10025 & 10294 \\ 
      \hline
      \textit{Cruise-500} & 4  & \textbf{Timeout} & 770 & 118239  & 120153 & 39425 & 40097 \\ 
      \hline
      \textit{AltLayer-50} & 3  & 15 & 9 & 3685  & 4147 & 1234 & 1395 \\ 
      \hline
      \textit{AltLayer-100} & 3  & 41 & 25 & 8885  & 9747 & 2968 & 3269 \\ 
      \hline
      \textit{AltLayer-150} & 3  & 153 & 100 & 30685  & 31947 & 10234 & 10699 \\ 
      \hline
      \textit{AltLayer-200} & 3  & \textbf{Timeout} & 269 & 52485  & 54147 & 17500 & 18064 \\    
      \hline
      \textit{Lift-5} & 3  & 1 & 2 & 310  & 884 & 122 & 355 \\ 
      \hline
      \textit{Lift-10} & 3  & 34 & 9 & 1585  & 4014 & 597 & 1522 \\ 
      \hline
      \textit{Lift-15} & 3  & \textbf{Timeout} & 162 & 4560  & 10844 & 1672 & 3989 \\ 
      \hline
      \textit{Lift-20} & 3  & \textbf{Timeout} &  789 & 6394 & 14948 & 3597 & 8255 \\  
      \hline
    \end{tabular}
    \caption{Comparision between \mobi and Monolithic}
    \label{fig:empirical-evaluation}
  \end{figure}

  \begin{figure}[b!]
    \begin{tabular}{@{}c@{}}
%
      \includegraphics[width=1.0\textwidth]{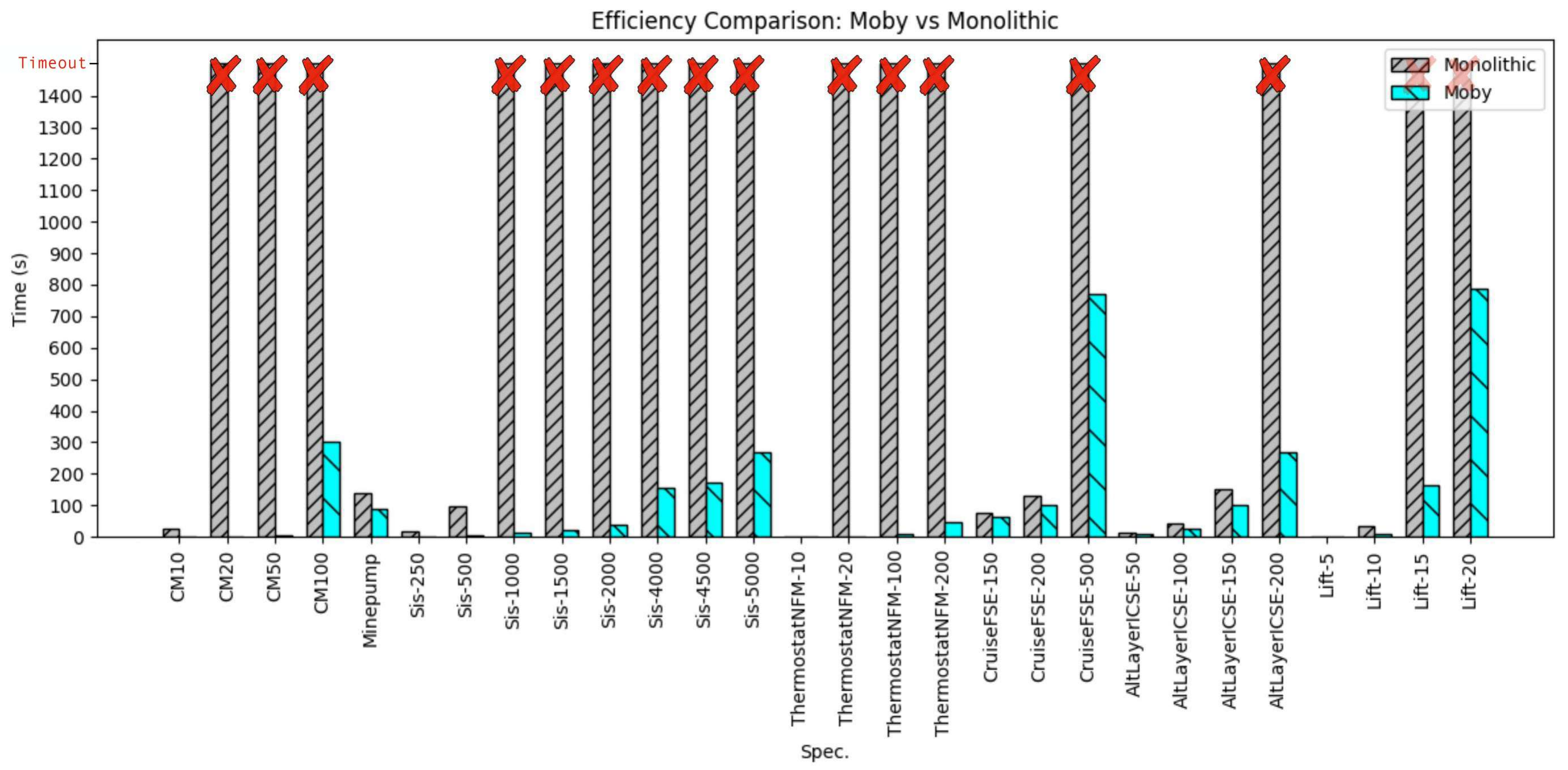} \\ 
      \includegraphics[width=1.0\textwidth]{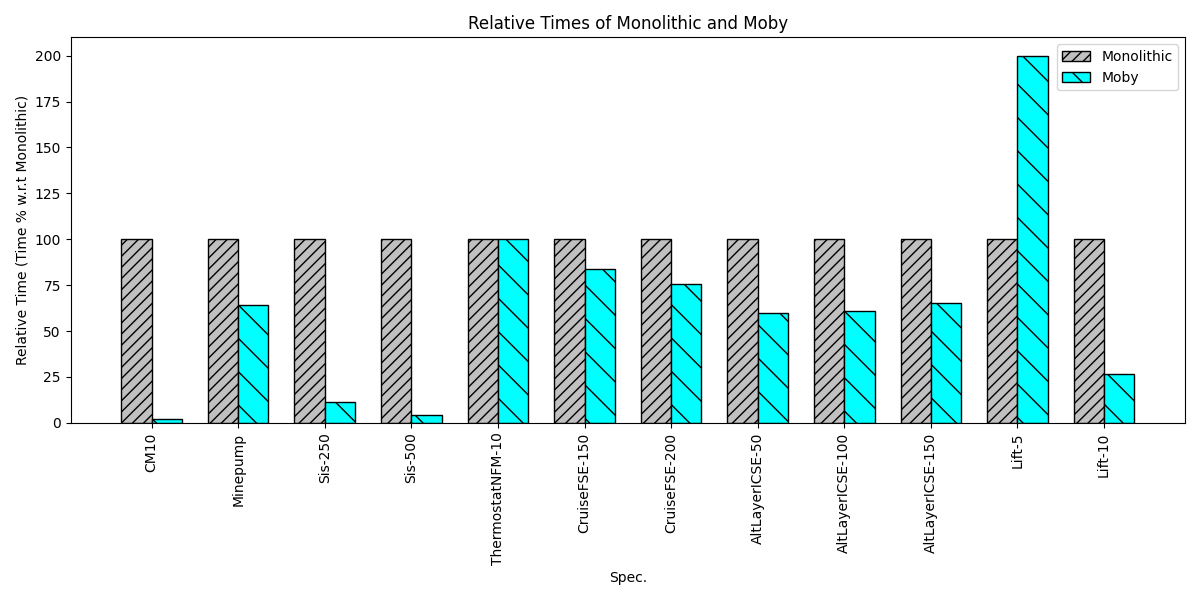} 
    \end{tabular}
    \vspace{-2em}
    \caption{Speed of \mobi vs monolithic synthesis. The figure above
      shows the time taken by a monolithic synthesis tool
      and the time taken by \mobi. The figure below normalizes
      the monolithic time to 100 for those that did not reach Timeout.}
    \label{fig:side_by_side}
  \end{figure}

To address them, we analyzed specifications from published literature,
evaluation of \emph{RE} tools, and case studies on SCR
specification and analysis:
\changing{
  \begin{compactitem}
  \item our counter machine running example \CM with varying bounds.
  \item Minepump: A mine pump
    controller~\cite{brizzio2023automated,carvalho2023acore,Degiovanni+2018,letier2008deriving},
    which manages a pump with sensors that detect high
    water levels and methane presence.
  \item Thermostat($n$): A thermostat~\cite{fifarek2017spear} that
    monitors a room temperature controls the heater and tracks heating
    duration.
  \item Lift($n$): A simple elevator controller for $n$ floors~\cite{SYNTCOMP}.
  \item Cruise($n$): A cruise control system~\cite{kirby1987example}
    which is in charge of maintaining the car speed on the
    occurrence of any event.
  \item Sis($n$): A safety injection system~\cite{degiovanni2018improving}, responsible for
    partially controlling a nuclear power plant by monitoring water
    pressure in a cooling subsystem.
  \item AltLayer($n$): A communicating state machine
    model~\cite{bultan2000action}.
  \end{compactitem}}
%

%
Fig.~\ref{tab:case-studies} shows the number of input/output
variables, assumptions (A), guarantees (G), and the number of modes
for each case.

\subsubsection{Experimental Results.}
To address \RQone we compare the size of the original specification
with the size of each projection measured by the number of clauses and
the formula length. 
\changing{To determine the formula's length, 
we adopt the methodologies outlined 
in~\cite{brizzio2023automated,carvalho2023acore}}.
Additionally, we compared the running time required for synthesizing
the original specification with the time taken for each projection,
note that we report the aggregated time taken to synthesize the
systems for each projection, when they can be solved independently and
in parallel to potentially improve efficiency.
The summarized results can be found in Fig.~\ref{fig:empirical-evaluation}.
We also provide additional insights in Fig.\ref{fig:side_by_side},
which highlights the significance of \mobi in enhancing the synthesis
time.
%

%
\changing{Our analysis demonstrates that \mobi successfully decomposes
  100\% of the specifications in our corpus, which indicates that
  \mobi is effective in handling complex specifications.
  Furthermore, \mobi consistently operates within the 25-minute
  timeout limit in all cases.
  In contrast, other relevant simultaneous decomposition
  methods~\cite{finkbeiner2022specification,filiot2010compositional}
  failed to decompose any of the specifications in our benchmark.
  This can be attributed to the intricate interdependencies between
  variables in our requirements, as elaborated in
  Section~\ref{sec:intro}.
  This observation not only supports the effectiveness of \mobi but
  also validates \RQtwo.

  Expanding on the impact of \mobi, our results show an average
  reduction of 64\% in specification size and a 65\% reduction in the
  number of clauses.
  These reductions underscore the advantages of employing \mobi in
  synthesizing implementations for LTL specifications that are beyond
  the capabilities of monolithic synthesizers.
  Additionally, \mobi's ability to achieve faster synthesis times for
  feasible specifications positions it as a compelling alternative to
  state-of-the-art synthesis tools. This suggests the validity of
  \RQthree.}


\section{Conclusion and Future Work}
\label{sec:conclusion}
We presented mode based decomposition for reactive
synthesis.
As far as we know, this is the first approach that exploits modes to
improve synthesis scalability.
Our method takes an LTL specification, along with a set of modes
representing different stages of execution, and a set of initial
conditions for each mode.
Our method computes projection for each mode ensuring that if all of
them are realizable, then the original specification is also
realizable.

We performed an empirical evaluation of an implementation of \mobi on
several specifications from the literature.
Our evaluation shows that \mobi successfully synthesizes
implementations efficiently, including cases for which monolithic
synthesis fails.
These results indicate that \mobi is effective for decomposing
specifications and can be used alongside other decomposition tools.

Even though modes are natural in RE, the need to
specify initial conditions is the major drawback of our technique.
We are currently investigating how to automatically compute the
initial conditions, using SAT based exploration.
We are also investigating the assessment of the quality of the
specifications generated using \mobi.



\clearpage

\vfill
\pagebreak

\bibliographystyle{splncs04}
\bibliography{short}

\vfill

\end{document}